\documentclass[submission,copyright,creativecommons]{eptcs}
\providecommand{\event}{AFL2014} 

\usepackage[utf8]{inputenc}
\usepackage{amsmath,amsthm,amssymb,amsfonts} 
\usepackage{enumitem} 
\usepackage{hyperref}
\usepackage{color}
\usepackage{subcaption}
\usepackage[english]{babel}

\usepackage{tikz}
\usepackage{tkz-graph}

\DeclareMathAlphabet{\mathcal}{OMS}{cmsy}{m}{n}

\title{Languages of lossless seeds}

\author{Karel Břinda
\institute{Laboratoire d’Informatique Gaspard Monge\\Université Paris-Est Marne-la-Vallée \\ Paris, France}
\email{karel.brinda@univ-mlv.fr}
}

\def\titlerunning{Languages of lossless seeds}
\def\authorrunning{K. Břinda}

\newcommand{\J}{\ensuremath{\texttt{-}}}
\newcommand{\F}{\ensuremath{\texttt{\#}}}
\newcommand{\setmin}{\backslash}
\newcommand{\uu}{\mathbf{u}}
\newcommand{\vv}{\mathbf{v}}
\newcommand{\ww}{\mathbf{w}}

\newcommand\A{\mathcal{A}}
\newcommand\Z{\mathbb{Z}}
\newcommand{\N}{\mathbb{N}}
\newcommand{\LL}{\mathcal{L}}

\newcommand{\sh}{\mathrm{sh}}
\newcommand{\lab}{\mathrm{lab}}
\newcommand{\seedset}[2]{\mathrm{Seed\,}^{#1}_{#2}}
\newcommand{\lkvalid}[2]{\mathrm{V\,}^{#1}_{#2}}
\newcommand{\compatible}[2]{\mathrm{C\,}^{#1}_{#2}}
\newcommand{\shift}{\sigma}
\newcommand{\emptyword}{\varepsilon}

\newcommand{\OO}{\mathcal{O}}

\newcommand{\seedOR}{{\oplus}}

\newcommand{\diag}[1]{\colorbox{black}{\color{white}{#1}}}
\newcommand{\alertPD}[1]{\colorbox{yellow}{#1}}
\newcommand{\alertPDb}[1]{\colorbox{red}{#1}}

\newtheorem{definition}{Definition}
\newtheorem{proposition}{Proposition}
\newtheorem{theorem}{Theorem}
\newtheorem{observation}{Observation}
\newtheorem{lemma}{Lemma}
\newtheorem{example}{Example}
\newtheorem{corollary}{Corollary}

\begin{document}
\maketitle


\begin{abstract}
Several algorithms for similarity search employ seeding techniques 
to quickly discard very dissimilar regions. In this paper, we study 
theoretical properties of lossless seeds, i.e., spaced seeds
having full sensitivity.
We prove that lossless seeds coincide with languages of certain 
sofic subshifts, hence they can be recognized by finite automata. 
Moreover, we show that these subshifts are fully given by the 
number of allowed errors $k$ and the seed margin $\ell$.
We also show that for a fixed $k$, optimal seeds must 
asymptotically satisfy $\ell \in \Theta(m^{\frac{k}{k+1}})$.
\end{abstract}


\section{Introduction}

The annual volume of data produced by the Next-Generation Sequencing technologies has been rapidly increasing;
even faster than growth of disk storage capacities.
Thus, new efficient algorithms and data-structures for processing, compressing and storing these data, are needed.

Similarity search represents the most frequent operation in bioinformatics. In huge DNA databases,
a two-phase scheme is the most widely used approach
to find all occurrences of a given string up to some Hamming or Levenshtein distance. 
First of all, most of dissimilar regions are discarded in
a fast \emph{filtration phase}.
Then, in a \emph{verification phase}, only ``hot candidates'' on similarity are processed by classical time-consuming algorithms like
Smith-Waterman \cite{smith-watermann} or
Needleman-Wunsch \cite{needleman-wunsch}.

Algorithms for the filtration phase are often based on so-called \emph{seed filters} which make use of the fact that two strings of the same length $m$ being in Hamming distance $k$ must necessarily share some exact patterns. 
These patterns are represented as strings over the alphabet $\{\F,\J\}$ called \emph{seeds}, where the ``matching'' symbol $\F$ corresponds to a matching position and the ``joker'' symbol $\J$ to a matching or a mismatching position. 

For instance, for two strings of length $15$, matching within two errors, shared patterns are, e.g.,
$\F\F\J\F\J\J\F\F\J\F$
or
$\F\F\F\F\F$.
For illustration, if we consider that two strings match as \texttt{===X=====X=====}
(where the symbols
\texttt{=} and \texttt{X}
represent respectively matching and mismatching positions), then the corresponding seed positions can be following:
\begin{center}
{
\texttt{===X=====X=====} \\
\texttt{.\F\F\J\F\J\J\F\F\J\F....}\\
\texttt{....\F\F\F\F\F......}
}
\end{center}
As the second seed is the longest possible contiguous seed in this 
case, we observe the main advantage of spaced seeds in comparison 
to contiguous seeds: for the same task, there exist spaced seeds 
with higher number of $\F$'s (so-called \emph{weight}).

Two basic characteristics of every seed are \emph{selectivity} and \emph{sensitivity}.
Selectivity measures
restrictivity of a filter created from the seed.
In general, higher weight
implies better selectivity of the filter. 
\emph{Lossless seeds} are those seeds having full sensitivity. They 
are easier to handle mathematically on one hand, but attain lower 
weight on the other hand.
\emph{Lossy seeds} are employed for practical purposes more since a small decrease in sensitivity can be compensated by considerable improvement of selectivity.

Nevertheless, only lossless seeds are considered in this paper.
For a given 
length $m$ of strings to be compared and a given number of allowed 
mismatches $k$ (such setting is called \emph{$(m,k)$-problem}), the 
aim is to design fully sensitive seeds with highest possible weight.

\subsection{Literature}

The idea of lossless seeds was originally introduced by Burkhardt 
and Karkkäinen~\cite{qgrampre,qgram}.
Let us remark that lossy spaced seed were used in the same 
time in the PatternHunter program \cite{patternhunter}.
Generalization of lossless seeds was studied by Kucherov 
et~al.~\cite{multiseed}.
given seed are required (the pattern is shared at more positions).
The authors also proved that, for a fixed number $k$ of mismatches, 
\emph{optimal seeds} (i.e., seeds with the highest possible weight 
among all seeds solving the given problem) must asymptotically satisfy
$m - w(m) \in \Theta(m^{\frac{k}{k+1}})$, where $w(m)$ denotes the 
maximal possible weight of a seed solving the $(m,k)$-problem.
They also started a systematic study of seeds created by repeating 
of short patterns.
Afterwards, the results on asymptotic properties of optimal seeds 
were generalized by
Farach-Colton et~al.~\cite{optimal_seeds_asymptotic}.
Computational complexity of optimal seed construction was derived by Nicolas and Rivals~\cite{hardnesspre,hardness}.

Further, the theory on lossless seed was significantly developed by Egidi and Manzini.
First, they studied seeds designed from mathematical objects called 
perfect rulers~\cite{perfect_rulers,perfect_rulers2}. The idea of 
utilization of some type of ``rulers'' was later independently 
extended
by KB~\cite{diplomka} (cyclic rulers) and, again, Edigi and Manzini~
\cite{periodic_seeds} (difference sets). In~\cite{periodic_seeds}, 
these ideas were extended also to seed families. Cyclic rulers and 
difference sets mathematically correspond to each other.
Edigi and Manzini~\cite{quadratic_residues} also showed possible 
usage of number-theoretical results on quadratic residues for seed 
design.

In practice, seeds often find their use in short-read mappers 
implementing hash tables (for more details on read mapping, see, 
e.g.,~\cite{Li2010a,Ribeca2012}). ZOOM~\cite{zoom} and
PerM~\cite{perm} are examples of mappers utilizing lossless seeds.

A list of papers on spaced seed is regularly maintained by Noé~\cite{seed_bibliography}.

\subsection{Our object of study}

One of the most important theoretical aspects of lossless seeds are their structural properties. 
Whereas good lossy seeds usually show irregularity,
it was observed that good lossless seeds are often repetitions of short patterns 
(\cite{multiseed,perm,diplomka,periodic_seeds}).
The question whether optimal seeds can be constructed in all cases by repeating patterns, which would be short with respect to seed length, remains open
(see~\cite[Conjecture~1]{diplomka}). Its answering would have  
practical impacts in development of bioinformatical software tools 
since the search space of programs for lossless seeds design could 
be significantly cut and also indexes in programs using lossless 
seeds for approximate string matching could be more memory 
efficient (like \cite{perm}).

\subsection{Paper organization and results}

In this paper, we follow and further develop ideas from~\cite{diplomka}. We concentrate on a parameter~$\ell$ called seed margin,
which is the difference between
the size~$m$ of compared strings and the length of a seed.

In Section~\ref{sec_preliminaries} we recall the notation used in combinatorics on words and symbolic dynamics.
In Section~\ref{sec_lossless_seeds} we formally define seeds and $(m,k)$-problems. Then we 
transform the problem of seed detection into another criterion (Theorem~\ref{thm_detection}) and also show asymptotic properties of $\ell$ for optimal seeds (Proposition~\ref{prop_asymptotic}).
In Section~\ref{sec_seed_subshifts} we prove that sets of seeds, obtained by fixing the parameters $k$ and $\ell$,
coincide with languages of some sofic subshifts. Therefore, those sets of seeds are recognized by finite automata. In Section~\ref{sec_application} we show applications of obtained results for seed design. These results provide a new view on lossless seeds and explain their periodic properties.


\section{Preliminaries}\label{sec_preliminaries}

Throughout the paper, we use a standard notation of combinatorics on words and symbolic dynamics.

\subsection{Combinatorics on words}

An \emph{alphabet} $\A=\{a_0,\ldots,a_{m-1}\}$ is a~finite set of symbols called \emph{letters}. In this paper, we will work exclusively with the alphabet $\{\F,\J\}$
A finite sequence of letters from $\A$ is called
a~\emph{finite word} (over $\A$).
The set $\A^{*}$ of all finite words (including the empty word~$\emptyword$)
provided with the operation of concatenation is a~free monoid. The concatenation is denoted multiplicatively.
If $w = w_0w_1\cdots w_{n-1}$ is a finite word over $\A$, we denote its 
length by $|w| = n$.
We deal also with
bi-infinite sequences
of letters from $\A$
called
{\em bi-infinite words} $\ww =\cdots \ww_{-2} \ww_{-1}| \ww_0 \ww_1 \ww_2\cdots$ over $\A$. The sets of all bi-infinite words over~$\A$ is denoted by $\A^{\Z}$.

A~finite word $w$ is called a~\emph{factor} of a word $\uu$ ($\uu$ being finite or bi-infinite) if there exist words $p$ and~$s$ (finite or one-side infinite) such that
$\uu=pws$.
For given indexes $i$ and $j$, the symbol $\uu[i,j]$ denotes the factor $\uu_i\uu_{i+1}\cdots\uu_j$ if $i\leq j$, or $\emptyword$ if $i>j$.
A~concatenation of $k$ words $w$
is denoted by $w^k$.
The set of all factors of a word $\uu$ ($\uu$ being finite or bi-infinite) is called the language of $\uu$ and denoted by $\LL(\uu)$. Its subset $\LL(\uu)\cap\A^n$ containing all factors of $\uu$ of length~$n$ is denoted by $\LL_n(\uu)$.

Let us remark that this notation will be used extensively in the whole text. For instance $\ww[2,5]\J^4$ denotes the word created by concatenation of the factor $\ww_2 \ww_3 \ww_4 \ww_5$ of a bi-infinite word $\ww$ and the word $\J\J\J\J$.
Similarly, for a finite word $v$ of length $n$, by $\cdots\J\J|v\J\J\cdots$ we denote the bi-infinite word $\uu$ such that for all
$i\in\{0,\ldots,n-1\}(\uu_i=w_i)$ and for all
$i\in\Z\setminus\{0,\ldots,n-1\}(\uu_i=\J)$.
For more information about combinatorics on words, we can refer to Lothaire~I~\cite{Lothaire1997}.

\subsection{Symbolic dynamics}\label{subsec_symdyn}

Consider an alphabet $\A$. On the set $\A^{\Z}$ of bi-infinite words over $\A$, we define a so-called Cantor metric~$d$ as
\begin{equation*}
	d(\uu,\vv) =
	\begin{cases}
		0 & \text{if } \uu=\vv, \\
		2^{-s} &  \text{if } \uu\not=\vv, \text{ where } s := \min\left\{|i|\ \big|\ \uu_i\not=\vv_i \right\}.
	\end{cases}
\end{equation*}
We define a \emph{shift} operation $\shift$ as
$[\shift(\uu)]_i = \uu_{i+1}$
for all $i\in\Z$.
The map $\sigma$ is invertible, and the power~$\sigma^k$ is defined by composition for all $k\in\Z$. The map $\sigma$ is continuous on $\A^\Z$, therefore, $(\A^\Z,\sigma)$ is a dynamical system, which is called a~\emph{full shift}.

A bi-infinite word $\uu\in\A^\Z$ \emph{avoids} a set of finite words $X$ if $\LL(\uu)\cap X=\emptyset$.
By $S_X$ we denote the set of all bi-infinite words that avoid $X$ and we call it a~\emph{subshift}. If $X$ is a regular language, $S_X$ is called \emph{sofic subshift}; if $X$ is finite, $S_X$ is called a \emph{subshift of finite type}.
The \emph{language} $\LL(S)$ of a subshift $S$ is the union of languages of all bi-infinite words from $S$.
By $\LL_n(S)$ we denote the set $\LL(S)\cap\A^{n}$.
It holds that a set $S \subseteq \A^{\Z}$ is a subshift if and only if
it is invariant under the shift map $\shift$ (that means $\shift(S)=S$)
and it is closed with respect to the Cantor metric.
A general theory of subshifts is well summarized in~\cite{Lind1995}.


\section{Lossless seeds}\label{sec_lossless_seeds}

In this section, we introduce basic definition formalizing lossless seeds. Then we introduce a parameter~$\ell$ called seed margin and show its asymptotic properties for optimal seeds.
Let us recall that $m$ denotes the length of strings to be compared and $k$ denotes the number of allowed mismatches.

\begin{definition}
The binary alphabet~$\A=\{\F,\J\}$ is called \emph{seed alphabet}. Every finite word over this alphabet is a \emph{seed}.
The \emph{weight} of a seed~$Q$ is the number of occurrences of the letter~$\F$ in~$Q$.
\end{definition}

\begin{definition}\label{def_(m,k)-problem}
	Let $m$ and $k$ be positive integers. Every set
	$\{i_1, \ldots, i_k\} \subseteq \{0,\ldots,m-1\}$
	is called \emph{error combination} of $k$ errors.
	
	Consider a seed $Q$ such that $|Q|<m$ and denote
	$\ell := m - |Q|$, which is the so-called
	\emph{seed margin}.
	Then $Q$
	\emph{detects} an error combination
	$\{i_1, \ldots, i_k\}\subseteq \{0,\ldots,m-1\}$ at position $t\in\{0,\ldots,\ell\}$ if 
	for all
	$j\in\{0,\ldots,|Q|-1\}$
	it holds
	$
		\left(
		Q_j = \F \implies j+t \not\in\{i_1,\ldots,i_k\}
		\right).
	$

	The seed $Q$ is said to \emph{solve} the $(m,k)$-problem if every error combination
	$\{i_1,\ldots,i_k\}\subseteq \{0,\ldots,m-1\}$	
	of $k$ errors is detected by $Q$ at some position $t\in\{0,\ldots,\ell\}$.
\end{definition}
Many combinatorial properties of seeds can be
studied from the perspective of bi-infinite words.
First, we need a seed analogy of the logical function \texttt{OR} applied on bi-infinite words and producing, again, a bi-infinite word.
\begin{definition}
	Consider $k$ bi-infinite words	$\uu^{(1)},\ldots,\uu^{(k)}$ over $\A$.
	We define a $k$-nary operation $\seedOR$ as
	\[\forall i\in\Z:
	\quad
		(\seedOR(\uu^{(1)},\ldots,\uu^{(k)}))_i =
		\begin{cases}
			\F & \text{if } (\uu^{(j)})_i = \F \enspace \text{ for some }
			j \in\{1,\ldots,k\}\\
			\J & \text{otherwise}
		\end{cases}
	\]
\end{definition}

The following theorem
will be crucial for seed analysis in the rest of the text.
It is mainly a translation of basic definitions to the
formalism of shifts and logical operations, but it enables us
to easily observe on which parameters (and how)
the structure of lossless seeds really depends.

\begin{theorem}\label{thm_detection}
	Let $m$ and $k$ be positive integers and $Q$ be a seed
	such that $|Q|<m$. Denote $\ell := m - |Q|$ and $\ww := \cdots\J\J|\J^\ell Q \J\J \cdots$.
	Then $Q$ detects an error combination
	$\{i_1,\ldots ,i_k\}\subseteq\{0,\ldots,m-1\}$
	at a position $t\in\{0,\ldots,\ell\}$ if and only if
	\begin{equation}\label{eq_basic_thm}
			\Bigl(\seedOR(\shift^{i_1}(\ww), \ldots, \shift^{i_k}(\ww)\Bigr)_{\ell-t} = \J.
	\end{equation}
\end{theorem}
	
\begin{proof}
	$Q$ detects $\{i_1,\ldots,i_k\}$ at position $t$ if 
	$
		\forall j\in\{0,\ldots,|Q|-1\} \bigl(
			Q_j = \F
			\implies
			j+t \not\in \{i_1,\ldots,i_k\}
			\bigr).
	$
	This is equivalent to $\forall p \in\{i_1,\ldots,i_k\}(\ww_{p - t + \ell}=\J)$,
	which is equivalent to~\eqref{eq_basic_thm}.
\end{proof}

\begin{corollary}\label{cor_basic_thm}
	$Q$ does not detect a combination 
	$\{i_1,\ldots ,i_k\}$
	at any position $t\in\{0,\ldots,\ell\}$	
	if and only if
	$
		(\seedOR(\shift^{i_1}(\ww), \ldots, \shift^{i_k}(\ww))
		[0,\ell] = \F^{\ell+1}.
	$
\end{corollary}

Let us mention that in the case of two errors,
Corollary~\ref{cor_basic_thm} corresponds to the 
Laser method~\cite[Section~4.1]{diplomka} (a JavaScript 
implementation is available at \cite{laser_web}) as we illustrate 
in the following example.

\begin{example}\label{ex_laser}
Consider a seed $Q = \F\F\J\F\J\J \J \J\J\F\J\F\F$ of length $14$
and the $(19,2)$-problem.
In Figure~\ref{fig_ex_k2} we show a corresponding schematic table. Denote $\ww := \cdots\J\J|\J^5 Q \J\J\cdots$.
The words 
$\seedOR(\shift_i(\ww),\shift_j(\ww))$ occur diagonally.
It is easily seen from Corollary~\ref{cor_basic_thm} that
$Q$ does not detect the error combination
$\{5,13\}$
since $\ell=5$ and $\bigl(\seedOR(\shift^5(\ww),\shift^{13}(\ww))\bigr)[0,5]=\F^{\ell+1}$.	
\end{example}

\begin{figure}[t] \centering
\begin{tabular}{p{0.12cm}p{0.12cm}||p{0.12cm}p{0.12cm}p{0.12cm}p{0.12cm}p{0.12cm}|p{0.12cm}p{0.12cm}p{0.12cm}p{0.12cm}p{0.12cm}p{0.12cm}p{0.12cm}p{0.12cm}p{0.12cm}p{0.12cm}p{0.12cm}p{0.12cm}p{0.12cm}p{0.12cm}|p{0.12cm}p{0.12cm}p{0.12cm}p{0.12cm}p{0.12cm}}
~&&$0$&$1$&$2$&$3$&$4$&$5$&$6$&$7$&$8$&$9$&$10$&$11$&$12$&$13$&$14$&$15$&$16$&$17$&$18$&&&&&\\ 
&&&&&&&\F&\F&\J&\F&\J&\J&\J&\J&\J&\J&\F&\J&\F&\F&&&&&\\ \hline\hline
$0$&&\diag{\J}&{\J}&{\J}&{\J}&{\J}&{\F}&{\F}&{\J}&{\F}&{\J}&{\J}&{\J}&{\J}&{\J}&{\J}&{\F}&{\J}&{\F}&{\F}&{\J}&{\J}&{\J}&{\J}&{\J}\\
$1$&&{\J}&\diag{\J}&{\J}&{\J}&{\J}&{\F}&{\F}&{\J}&{\F}&{\J}&{\J}&{\J}&{\J}&{\J}&{\J}&{\F}&{\J}&{\F}&{\F}&{\J}&{\J}&{\J}&{\J}&{\J}\\
$2$&&{\J}&{\J}&\diag{\J}&{\J}&{\J}&{\F}&{\F}&{\J}&{\F}&{\J}&{\J}&{\J}&{\J}&{\J}&{\J}&{\F}&{\J}&{\F}&{\F}&{\J}&{\J}&{\J}&{\J}&{\J}\\
$3$&&{\J}&{\J}&{\J}&\diag{\J}&{\J}&{\F}&{\F}&{\J}&{\F}&{\J}&{\J}&{\J}&{\J}&{\J}&{\J}&{\F}&{\J}&{\F}&{\F}&{\J}&{\J}&{\J}&{\J}&{\J}\\
$4$&&{\J}&{\J}&{\J}&{\J}&\diag{\J}&{\F}&{\F}&{\J}&{\F}&{\J}&{\J}&{\J}&{\J}&{\J}&{\J}&{\F}&{\J}&{\F}&{\F}&{\J}&{\J}&{\J}&{\J}&{\J}\\\hline
$5$&\F&{\F}&{\F}&{\F}&{\F}&{\F}&\diag{\F}&{\F}&{\F}&{\F}&{\F}&{\F}&{\F}&{\F}&{\F}&{\F}&{\F}&{\F}&{\F}&{\F}&{\F}&{\F}&{\F}&{\F}&{\F}\\
$6$&\F&{\F}&{\F}&{\F}&{\F}&{\F}&{\F}&\diag{\F}&{\F}&{\F}&{\F}&{\F}&{\F}&{\F}&{\F}&{\F}&{\F}&{\F}&{\F}&{\F}&{\F}&{\F}&{\F}&{\F}&{\F}\\
$7$&\J&{\J}&{\J}&{\J}&{\J}&{\J}&{\F}&{\F}&\diag{\J}&{\F}&{\J}&{\J}&{\J}&{\J}&{\J}&{\J}&{\F}&{\J}&{\F}&{\F}&{\J}&{\J}&{\J}&{\J}&{\J}\\
$8$&\F&{\F}&{\F}&{\F}&{\F}&{\F}&{\F}&{\F}&{\F}&\diag{\F}&{\F}&{\F}&{\F}&{\F}&{\F}&{\F}&{\F}&{\F}&{\F}&{\F}&{\F}&{\F}&{\F}&{\F}&{\F}\\
$9$&\J&{\J}&{\J}&{\J}&{\J}&{\J}&{\F}&{\F}&{\J}&{\F}&\diag{\J}&{\J}&{\J}&{\J}&{\J}&{\J}&{\F}&{\J}&{\F}&{\F}&{\J}&{\J}&{\J}&{\J}&{\J}\\
$10$&\J&{\J}&{\J}&{\J}&{\J}&{\J}&{\F}&{\F}&{\J}&{\F}&{\J}&\diag{\J}&{\J}&{\J}&{\J}&{\J}&{\F}&{\J}&{\F}&{\F}&{\J}&{\J}&{\J}&{\J}&{\J}\\
$11$&\J&{\J}&{\J}&{\J}&{\J}&{\J}&{\F}&{\F}&{\J}&{\F}&{\J}&{\J}&\diag{\J}&{\J}&{\J}&{\J}&{\F}&{\J}&{\F}&{\F}&{\J}&{\J}&{\J}&{\J}&{\J}\\
$12$&\J&{\J}&{\J}&{\J}&{\J}&{\J}&{\F}&{\F}&{\J}&{\F}&{\J}&{\J}&{\J}&\diag{\J}&{\J}&{\J}&{\F}&{\J}&{\F}&{\F}&{\J}&{\J}&{\J}&{\J}&{\J}\\
$13$&\J&{\J}&{\J}&{\J}&{\J}&{\J}&\alertPDb{\F}&{\F}&{\J}&{\F}&{\J}&{\J}&{\J}&{\J}&\diag{\J}&{\J}&{\F}&{\J}&{\F}&{\F}&{\J}&{\J}&{\J}&{\J}&{\J}\\
$14$&\J&{\J}&{\J}&{\J}&{\J}&{\J}&{\F}&\alertPD{\F}&{\J}&{\F}&{\J}&{\J}&{\J}&{\J}&{\J}&\diag{\J}&{\F}&{\J}&{\F}&{\F}&{\J}&{\J}&{\J}&{\J}&{\J}\\
$15$&\F&{\F}&{\F}&{\F}&{\F}&{\F}&{\F}&{\F}&\alertPD{\F}&{\F}&{\F}&{\F}&{\F}&{\F}&{\F}&{\F}&\diag{\F}&{\F}&{\F}&{\F}&{\F}&{\F}&{\F}&{\F}&{\F}\\
$16$&\J&{\J}&{\J}&{\J}&{\J}&{\J}&{\F}&{\F}&{\J}&\alertPD{\F}&{\J}&{\J}&{\J}&{\J}&{\J}&{\J}&{\F}&\diag{\J}&{\F}&{\F}&{\J}&{\J}&{\J}&{\J}&{\J}\\
$17$&\F&{\F}&{\F}&{\F}&{\F}&{\F}&{\F}&{\F}&{\F}&{\F}&\alertPD{\F}&{\F}&{\F}&{\F}&{\F}&{\F}&{\F}&{\F}&\diag{\F}&{\F}&{\F}&{\F}&{\F}&{\F}&{\F}\\
$18$&\F&{\F}&{\F}&{\F}&{\F}&{\F}&{\F}&{\F}&{\F}&{\F}&{\F}&\alertPD{\F}&{\F}&{\F}&{\F}&{\F}&{\F}&{\F}&{\F}&\diag{\F}&{\F}&{\F}&{\F}&{\F}&{\F}\\\hline
&&{\J}&{\J}&{\J}&{\J}&{\J}&{\F}&{\F}&{\J}&{\F}&{\J}&{\J}&{\J}&{\J}&{\J}&{\J}&{\F}&{\J}&{\F}&{\F}&\diag{\J}&{\J}&{\J}&{\J}&{\J}\\
&&{\J}&{\J}&{\J}&{\J}&{\J}&{\F}&{\F}&{\J}&{\F}&{\J}&{\J}&{\J}&{\J}&{\J}&{\J}&{\F}&{\J}&{\F}&{\F}&{\J}&\diag{\J}&{\J}&{\J}&{\J}\\
&&{\J}&{\J}&{\J}&{\J}&{\J}&{\F}&{\F}&{\J}&{\F}&{\J}&{\J}&{\J}&{\J}&{\J}&{\J}&{\F}&{\J}&{\F}&{\F}&{\J}&{\J}&\diag{\J}&{\J}&{\J}\\
&&{\J}&{\J}&{\J}&{\J}&{\J}&{\F}&{\F}&{\J}&{\F}&{\J}&{\J}&{\J}&{\J}&{\J}&{\J}&{\F}&{\J}&{\F}&{\F}&{\J}&{\J}&{\J}&\diag{\J}&{\J}\\
&&{\J}&{\J}&{\J}&{\J}&{\J}&{\F}&{\F}&{\J}&{\F}&{\J}&{\J}&{\J}&{\J}&{\J}&{\J}&{\F}&{\J}&{\F}&{\F}&{\J}&{\J}&{\J}&{\J}&\diag{\J}\\
\end{tabular}

\caption{The Laser method for the $(19,2)$-problem and the seed $Q = \F\F\J\F \J\J\J\J\J\J  \F\J\F\F$ in Example~\ref{ex_laser}.
}
\label{fig_ex_k2}
\end{figure}

Even though the basic parameters in the concept of
$(m,k)$-problems are $m$ and $k$;
as follows from Theorem~\ref{thm_detection},
the parameters determining structure of seeds are~$\ell$ and~$k$. 
Therefore, in the next section, we will fix them and study seeds $Q$ solving
$(|Q|+\ell,k)$-problems. Hence, when increasing $m$,
the seed must be extended in order to keep $\ell$ constant.

To complete this section, we show the asymptotic relation of
$m$, $\ell$, and $k$ for optimal seeds.
Let us fix the parameter $k$.
Let $w(m)$ denote the maximal weight of a seed solving the $(m,k)$-problem.
It was proved in~\cite[Lemma~4]{multiseed}
that $m-w(m)\in\Theta(m^{\frac{k}{k+1}})$.
We show that $\ell$ has the same asymptotic behavior.

\begin{proposition}\label{prop_asymptotic}
	Let $k$ be a fixed positive integer and $w(m)$ denote the maximal weight of a seed solving
	the $(m,k)$-problem.
	Let $H(m)$ be the set of all seeds with weight $w(m)$ solving the $(m,k)$-problem.
	For every positive $m$, set $\ell(m) := m - |Q|$, where $Q$ is an arbitrary seed from $H(m)$. Then
	$
		\ell(m) \in \Theta (m - w(m)).
	$
\end{proposition}

\begin{proof}
	Since $m \geq \ell(m) + w(m)$, we get trivially the
	upper bound as $\ell(m)\in\OO(m - w(m))$.
	Now let us prove the lower bound. Let $k$ be fixed.
	Since $m-w(m)\in\Theta(m^{\frac{k}{k+1}})$ for
	optimal seeds, it also holds that $(m-w(m))^{k+1} \in \OO(m^k)$. From combinatorial considerations on seed detection, we get
	$
		{m \choose k} \leq {m - w(m) \choose k}(\ell+1).
	$
	By combining the last two formulas, we obtain $\ell(m) \in \Omega(m-w(m))$, which concludes the proof.
\end{proof}


\section{Seed subshifts}\label{sec_seed_subshifts}

In this section, we show the relation between lossless seeds
and subshifts.
First, we denote sets of seeds obtained by fixing
the parameters $\ell$ and $k$.
Afterwards, we prove that they coincide with languages of
certain sofic subshifts. After defining
functions checking the criterion given by
Corollary~\ref{cor_basic_thm} globally on bi-infinite words,
we show that the subshift are exactly the sets of
bi-infinite words, for which these functions have
the upper bound~$\ell$.

\begin{definition}\label{def_S_k(ell)}
	Let $\ell$ and $k$ be positive integers. 
	The set of all seeds such that each seed $Q$ 
	solves the $(|Q|+\ell,k)$-problem
	is denoted by $\seedset{\ell}{k}$.
\end{definition}
\begin{example}\label{ex_1}
	$\seedset{3}{2} = \{\emptyword,\F,\J, \F\J, \J\F, \J\J ,\F\J\J, \J\F\J, \J\J\F, \J\J\J, \F\J\J\F, \F\J\J\J, \J\F\J\J, \J\J\F\J, \J\J\J\F, \J\J\J\J, \ldots\}$.
\end{example}


\subsection[Functions $\sh_k$ and $(l,k)$-valid bi-infinite words]{Functions $\sh_k$ and $(\ell,k)$-valid bi-infinite words}

\begin{definition}\label{def_sh}
	Consider a positive integer $k$. We define a function
	$\sh_k:\ (\A^\Z)^k \to \N_0\cup\{+\infty\}$ as:
	\begin{equation}\label{eq_def_sh}
		\sh_k(\uu^{(1)},\ldots,\uu^{(k)}) =
		\sup_{i_1,\ldots,i_k\in\Z~}
		\sup_{p\in\N_0}\left\{
			p\ |\ \vv
			[0,p-1] = \F^p, \text{ where }
			\vv = 
			\seedOR \left(
			\shift^{i_1}(\uu^{(1)}),\ldots,\shift^{i_k}(\uu^{(k)})
			\right)
		\right\}.
	\end{equation}
	
We extend the range of the function
$\sh_k(\cdot,\ldots,\cdot)$ to $(\A^*)^k$.
Finite words $w$ are transformed into bi-infinite words $\vv$ as $\vv:=\cdots\J\J|w\J\J\cdots$.
\end{definition}

Informally said, $sh_k(\uu^{(1)},\ldots,\uu^{(k)})$ is equal to
\begin{itemize}
	\item
		a finite $s\in\N_0$ if after arbitrary ``aligning'' of the words followed by the logical \texttt{OR} operation (in the Laser method the diagonal bi-infinite words), each run of $\F$'s has length at most $s$ and the value $s$ is attained for some ``alignment'';
	\item
		$+\infty$ if there exists an ``alignment'' with run of infinitely many $\F$'s (e.g., $\sh_2(\cdots vv|vv\cdots,\cdots ww|ww \cdots)$ with $v=\F\F\J$ and $w=\F\J\J$).
\end{itemize}

Every function $sh_k$ is symmetric and shift invariant with respect to all variables. The following observations show how to estimate
their values for given $k$~bi-infinite words.

\begin{observation}[Lower estimate]\label{obs_shifts_fact}
	Let
	$\uu^{(1)},\ldots,\uu^{(k)}$
	be bi-infinite words.
	If	$\seedOR(\shift^{i_1}(\uu^{(1)}),\ldots,\shift^{i_k}(\uu^{(k)}))$
		has a factor $\F^p$
		for some $i_1,\ldots,i_k$\/; then
		$
			\sh_k(\uu^{(1)},\ldots,\uu^{(k)})\geq p.
		$
\end{observation}

\begin{observation}[Upper estimate]\label{obs_ord}
	Let
	$\uu^{(1)},\ldots,\uu^{(k)},\vv^{(1)},\ldots,\vv^{(k)}$
	be bi-infinite words such that
	$\uu^{(1)} \preceq \vv^{(1)},\ldots,\uu^{(k)} \preceq \vv^{(k)}$,
	where $\preceq$ is a relation defined as
	\begin{equation}\label{eq_preceq}
		\uu \preceq \vv
		\qquad \iff \qquad
		\left(
			\uu_i = \F \implies \vv_i = \F
		\right)
		\enspace
		\text{holds for all } i\in\Z.
	\end{equation}
	Then $\sh_k(\uu^{(1)},\ldots,\uu^{(k)}) \leq \sh_k(\vv^{(1)},\ldots,\vv^{(k)})$.
\end{observation}

Bi-infinite words for which the $sh_k$ function is bounded by some $\ell$, will be the ``bricks'' of our subshifts. Their factors $Q$ are exactly those seeds solving $(|Q|+\ell,k)$-problems. 
\begin{definition}\label{def_ell_valid}
	A bi-infinite word $\uu$ satisfying
	$\sh_k(\uu,\ldots,\uu) \leq \ell$
	is called an \emph{$(\ell;k)$-valid bi-infinite word}.
	For fixed positive integers $\ell$ and $k$,
	we denote the set of all $(\ell;k)$-valid words
	by~$\lkvalid{\ell}{k}$.
\end{definition}

\begin{lemma}\label{lem_(l,k)-valid_cor_solving_seeds}
	A seed $Q$ solves the $(|Q|+\ell,k)$-problem if and only if it is a factor of an $(\ell,k)$-valid bi-infinite word.
\end{lemma}
\begin{proof}
	\begin{enumerate}
	\item[$\implies$:]
		The word
		$\ww := \cdots \J\J | \J^\ell Q \J\J\cdots$
		must be $(\ell,k)$-valid since otherwise $Q$ would not solve the $(|Q|+\ell,k)$-problem by Corollary~\ref{cor_basic_thm}.
		
	\item[$\impliedby$:]
			For a contradiction assume that there exists a factor~$Q$ of a bi-infinite word~$\uu$,
			which does not solve the $(|Q| + \ell,k)$-problem.
			Let the non-detected error combination be
			$\{i_1,\ldots,i_k\}$.
			Denote $\ww = \cdots\J\J| \J^\ell Q \J\J\cdots$.
	
			We use shift invariance of $\sh_k$ and Observation~\ref{obs_ord} to get
			\begin{equation}\label{eq_contr_vv_Q}
				\sh_k(\ww,\ldots,\ww) \leq \sh_k(\uu,\ldots,\uu) \leq \ell.
			\end{equation}
			Since $Q$ does not detect the error combination
			$\{i_1,\ldots,i_k\}$, it follows from Corollary~\ref{cor_basic_thm} that
			\[
				(\seedOR(\shift^{i_1}(\ww), \ldots, \shift^{i_k}(\ww))[0,\ell] = \F^{\ell+1}.
			\]			
			Nevertheless, this gives us a lower estimate on
			$\sh_k(\ww,\ldots,\ww)$, which is
			contradicting~\eqref{eq_contr_vv_Q}.\qedhere
	\end{enumerate}
\end{proof}


\subsection[subshifts of $(l,k)$-valid words]{Subshifts of $(\ell,k)$-valid words}

The property of $(\ell,k)$-validity is preserved under the shift operation. 
Moreover, the sets $\lkvalid{\ell}{k}$ of $(\ell,k)$-valid words are subshifts.
To prove it, we need to find a criterion for verifying $(\ell,k)$-validity based on comparing finite factors of a given bi-infinite word.
	
\begin{lemma}\label{lem_pairs}
	Let $\uu$ be a bi-infinite word over the seed alphabet $\A$.
	Then the following statements are equivalent:
	\begin{enumerate}
		\item $\uu$ is $(\ell;k)$-valid;
		\item $\forall v^{(1)},\ldots,v^{(k)} 
				\in \LL_{\ell+1}(\uu)\bigl(
					\sh_k (v^{(1)},\ldots,v^{(k)}) \leq \ell
				\bigr)$;
		\item $\forall w^{(1)},\ldots,w^{(k)} 
				\in \LL_{\ell+1}(\uu)\bigl(
					\seedOR (w^{(1)},\ldots,w^{(k)}) \not= \F^{\ell+1}
				\bigr)$.
	\end{enumerate}
\end{lemma}
\begin{proof}
	We prove three implications.
	 
	\begin{enumerate}[leftmargin=1.4cm]
		\item[1$\implies$2:]
			Consider any such factors $v^{(1)},\ldots,v^{(k)}$. Find their positions $i_1,\ldots,i_k$ in $\uu$.
				It holds that
				\[
					\cdots\J\J|v^{(1)}\J\J\cdots \preceq \shift^{i_1}(\uu), \enspace
					\ldots, \enspace
					\cdots\J\J|v^{(k)}\J\J\cdots \preceq \shift^{i_k}(\uu),
				\]
				where $\preceq$ is the relation defined
				by~\eqref{eq_preceq}.
				By combining the assumption, shift invariance of $\sh_k$, and Observation~\ref{obs_ord}, we obtain
				$
					\sh_k (v^{(1)},\ldots,v^{(k)})
					\leq
					\sh_k(\uu,\ldots,\uu)
					\leq 
					\ell
				$.
		
		\item[2$\implies$3:]
			It is an easy consequence of the definition of the $\sh_k$ function.
		
		\item[3$\implies$1:]
			For a contradiction assume that $\uu$ is not $(\ell,k)$-valid. Then
			there exist integers $i_1,\ldots,i_k$ such that
			$\seedOR ( \shift^{i_1}(\uu), \ldots, \shift^{i_1}(\uu))[0,\ell] = \F^{\ell+1}$.
			\qedhere
	\end{enumerate}
\end{proof}


The main consequence of Lemma~\ref{lem_pairs} is the fact that
every seed must be constructed from reciprocally compatible tiles of length $\ell+1$. To describe this property, we define a relation of compatibility on the set $\A^{\ell+1}$.

\begin{definition}\label{def_compatible}
	For given positive integers $\ell$ and $k$, we define the $k$-nary \emph{compatibility relation} $\compatible{\ell}{k}$ on $\A^{\ell+1}$ as
	\[
		\compatible{\ell}{k}(v^{(1)},\ldots,v^{(k)})
		\quad \iff \quad
		\sh_k(v^{(1)},\ldots,v^{(k)})\leq \ell.
	\]
\end{definition}

\begin{corollary}\label{cor_comp}
	Let $\uu$ be a bi-infinite word over the seed alphabet $\A$.
	The word $\uu$ is $(\ell,k)$-valid if and only if
	\enspace
	$
		\forall v^{(1)},\ldots,v^{(k)}\in\LL_{\ell+1}(\uu)
		\left(
			\compatible{\ell}{k}(v^{(1)},\ldots,v^{(k)})
		\right).
	$
\end{corollary}

Now let us prove that $(\ell,k)$-valid words really form subshifts. 
We only need to show that $(\ell,k)$-valid words are exactly
those words, which can be
created from compatible ``tiles''.
\begin{lemma}\label{lem_lkvalid=shift_space}
	Let $\ell$ and $k$ be positive integers.
	The set $\lkvalid{\ell}{k}$ of all $(\ell,k)$-valid words
	is a subshift.
\end{lemma}
\begin{proof}
	We prove the lemma by construction of a set
	$X$ of forbidden words (as they are introduced
	in~\ref{subsec_symdyn}).
	Take
	\[
		X := \bigl\{
			x\in\A^*
			\ |\ \exists 
				v^{(1)},\ldots,v^{(k)}
			\in \LL_{\ell+1}(x)
			\bigl(
				\lnot \compatible{\ell}{k}(v^{(1)},\ldots,v^{(k)})
			\bigr)
		\bigr\}.
	\]
	The set $X$ contains all possible finite words having
	some factors, which are ``incompatible'' with respect to
	the given~$\ell$ and~$k$.
	Hence, the subshift $S_X$ contains exactly all
	bi-infinite words $\uu$ satisfying
	$\forall v_1,\ldots,v_k\in\LL_{\ell+1}(\uu)
	\bigl( \compatible{\ell}{k} (v_1,\ldots,v_k) \bigr)$
	and we obtain $S_X = \lkvalid{\ell}{k}$ by
	Corollary~\ref{cor_comp}.
\end{proof}

\begin{example}\label{ex_nonexsingle}
	Even though both of the seeds
	$Q^{(1)} = \F\F\J\F\J\J$ and
	$Q^{(2)} = \J\J\F\J\F\F$ solve the $(11,2)$-problem,
the seed $Q = Q^{(1)} \J\J Q^{(2)}$ does not solve the
$(19,2)$-problem as we have seen in Example~\ref{ex_laser}. Since
$Q^{(1)} \oplus Q^{(2)} = \F^6$, any seed $\tilde{Q}$ of the form $\tilde{Q} = Q^{(1)} \J^p Q^{(2)}$ cannot solve the $(|\tilde{Q}|+5,2)$-problem.
\end{example}



It follows from the last example that the subshift
$\lkvalid{5}{2}$ of all $(5,2)$-valid words
is not of finite type.
Nevertheless, every subshift $\lkvalid{\ell}{k}$ 
must be a union of subshifts of finite type,
which can be constructed from so-called $(\ell,k)$-generating sets.

\begin{definition}\label{def_gen_set_nemax}
	For given positive integers $\ell$ and $k$, a subset $G$ of $\A^{\ell+1}$ is called \emph{$(\ell,k)$-generating set} if the following conditions are satisfied:
	\begin{enumerate}
		\item for all $v^{(1)},\ldots,v^{(k)} \in G$, it holds
			$
				\compatible{\ell}{k}(v^{(1)},\ldots,v^{(k)});
			$
		\item it is maximal possible (i.e., it cannot contain any other word from $\A^{\ell+1}$).
	\end{enumerate}
\end{definition}

\begin{observation}
Let us take a word from 
an $(\ell,k)$-generating set~$G$. If we remove the last or the first letter and concatenate the letter~$\J$ to the beginning or to the end of the word, we obtain again a word from~$G$. Therefore,
every $(\ell,k)$-generating set $G$~must contain, e.g., the
word~$\J^{\ell+1}$.
\end{observation}

Every generating set $G$ fully determines a subshift of finite type, we will denote it by $S(G)$. This subshift contains all bi-infinite words $\uu$ such that $\LL_{\ell+1}(\uu)\subseteq G$.


\begin{definition}\label{def_generating_set}
	Consider a seed $Q$ and an $(\ell,k)$-generating set $G$.
	By $S(G)$, we denote
	the subshift $S_X$ of finite type given by $X = \A^{\ell+1}\setmin G$.	
	We say that a seed $Q$ is \emph{generated} by $G$ if $Q\in \LL(S(G))$.
\end{definition}

In other words, a seed $Q$ satisfying $|Q| \geq \ell+1$ is generated by $G$ if $\LL_{\ell+1}(Q)\subseteq G$.
A seed $Q$ such that $|Q| < \ell+1$ is generated by $G$ if
$\exists w \in G \bigl(Q \in \LL(w)\bigr)$. We can also observe 
that every $(\ell,k)$-valid word is generated by some
$(\ell,k)$-generating set.

\begin{observation}\label{obs_generating_set_existence}
	For every $(\ell,k)$-valid bi-infinite word $\uu$, there exists an $(\ell,k)$-generating set $G$ such that $\uu\in S(G)$.
\end{observation}

\begin{example}\label{ex_graph}
	Continue with the setting from Example~\ref{ex_1}. Consider the only one $(3,2)$-generating set
	$G = \{\F\J\J\F, \F\J\J\J, \J\F\J\J, \J\J\F\J, \J\J\J\F, \J\J\J\J\}$. 
	Since $S(G)$ is of finite type, it follows from theory of symbolic dynamics that there exists a strongly connected labeled graph $H$ such that
	$S(G)$ coincide with labels of all bi-infinite paths in
	$H$ (for details, see~\cite{Lind1995}).
	This graph also determines a finite automaton recognizing the set $\LL(S(G))$, i.e., the set of labels of finite paths in $H$.
	Such automaton can be created from a de-Bruijn graph.
	However, it would not be minimal
	as it is shown in Figure~\ref{fig_deBruijn}.
\end{example}

\begin{figure}[t]\centering
\begin{subfigure}[b]{0.4\textwidth}
\centering
\begin{tikzpicture}[
	->,
	>=stealth',
	shorten >=1pt,
	auto,
	node distance=2.5cm,
	main node/.style={fill=blue!20,draw}
]

  \node[main node] (6)  {\J\J\J\J};
  \node[main node] (2) [below left of=6] {\F\J\J\J};
  \node[main node] (5) [below right of=6] {\J\J\J\F};
  \node[main node] (3) [below of=2] {\J\F\J\J};
  \node[main node] (4) [below of=5] {\J\J\F\J};
  \node[main node] (1) [below right of=3] {\F\J\J\F};

  \path[every node/.style={font=\sffamily\small}]
    (1) edge node [] {\J} (4)
    (2) edge node [] {\F} (5)
    (2) edge node [] {\J} (6)
    (3) edge node [] {\J} (2)
    (3) edge node [] {\F} (1)
    (4) edge node [] {\J} (3)
    (5) edge node [] {\J} (4)
    (6) edge node [] {\F} (5)
    (6) edge [loop above] node {\J} (6);
\end{tikzpicture}
\caption{A graph created as a de-Bruijn graph from the set of vertices $G$.}
\label{subfig_deBruijn}
\end{subfigure}
\qquad
\begin{subfigure}[b]{0.4\textwidth}
\centering
\begin{tikzpicture}[
	->,
	>=stealth',
	shorten >=1pt,
	auto,
	node distance=2.5cm,
	main node/.style={fill=blue!20,draw}
]

  \node[main node] (1)  {1};
  \node[main node] (2) [below right of=1] {2};
  \node[main node] (3) [below left of=1] {3};

  \path[every node/.style={font=\sffamily\small}]
    (1) edge node [] {\F} (2)
    (2) edge node [] {\J} (3)
    (3) edge node [] {\J} (1)
    (1) edge [loop above] node {\J} (1);
\end{tikzpicture}
\caption{The previous graph after minimization.}
\label{subfig_minimized}
\end{subfigure}

\caption{
Labeled graphs $H$ for the subshift $S(G)$ in Example~\ref{ex_graph}.}
		\label{fig_deBruijn}
	\end{figure}

\begin{theorem}\label{thm_S_k(l)_regular}
	Let $k$	and $\ell$ be positive integers.
	The set $\seedset{\ell}{k}$ is a regular language.
\end{theorem}

\begin{proof}
	There can be only finite number of $(\ell,k)$-generating sets; denote them $G_1, \ldots, G_d$.
	It follows from Observation~\ref{obs_generating_set_existence}
	that
	$S(G_1)\cup \ldots \cup S(G_d)  = \lkvalid{\ell}{k}$
	and, from Lemma~\ref{lem_(l,k)-valid_cor_solving_seeds}, we know that
	$\LL(\lkvalid{\ell}{k}) = \seedset{\ell}{k}$.	
	
	For every $i\in\{1,\ldots,d\}$, the set $S(G_i)$ is a subshift of finite type, so every set $\LL(S(G_i))$ is a regular language.
	Since the set $\seedset{\ell}{k}$ is a union of finitely many regular languages, it is a regular language.	
\end{proof}

\section{Application for seed design}\label{sec_application}

In this section, we describe how to design seeds with knowledge of
an $(\ell,k)$-generating set. Then we show how to search $(\ell,1)$ and $(\ell,2)$-generating sets.

\subsection{Seed design using generating sets}

Let us have an $(\ell,k)$-generating set~$G$
and let us consider a task of designing a seed~$Q$
of length~$s$, which would solve the $(\ell+s,k)$-problem.
If $s\leq\ell+1$, we can take an arbitrary 
factor of length $s$ of any word from $G$.

If $s>\ell+1$, we need to construct the seed in $s - \ell$~steps by extending letter by letter.
In the first step, we take an arbitrary word~$w\in G$
and set~$Q:=w$. In every other step, we take any
word~$w$ from~$G$ such that the last~$\ell$ letters of~$Q$ are
equal to $\ell$ first letters of~$w$
and concatenate the last letter of~$w$ to~$Q$. Existence of such word~$w$ is guaranteed since we can use at least the letter $\J$ in every step.

\subsection{Generating sets for $k=1$}

As a simple consequence of Corollary~\ref{cor_basic_thm}, we get a full characterization of all seeds solving $(m,1)$-problems (\cite[Theorem~5]{diplomka}).
\begin{proposition}\label{prop_(m,1)}
	$
		\seedset{\ell}{1}
		= \{Q \in \A^*\ |\ \F^{\ell+1} \text{ is not a factor of } Q \}
	$
\end{proposition}
\begin{proof}
	Denote $\vv = \cdots\J\J|\J^\ell Q\J\J\cdots$ and $\ell=m-|Q|$. Then from Corollary~\ref{cor_basic_thm}
	follows that:
	$Q$ solves the $(m,1)$-problem $\iff$
	$\forall i\in\Z \left(
		(\shift^{i}(\vv))[0,\ell]\not= \F^{\ell+1}
		\right)$
	$\iff$
	$Q$ does not contain $\F^{\ell+1}$.
\end{proof}

Thus, for every positive $\ell$, the only
$(\ell,1)$-generating set
is $\A^{\ell+1} \setminus \{\F^{\ell+1}\}$, i.e., the set of all words of length $\ell+1$ except $\F^{\ell+1}$.

\subsection{Generating sets for $k=2$}

Let $k=2$ and $\ell$ be an arbitrary fixed positive integer.
We can derive all $(\ell,2)$-generating sets using graph theoretical
methods by transformation to independent sets search.
Let $V := \{w^{(1)},\ldots,w^{(q)}\}$ denote the set 
of all seeds of length $\ell+1$ solving the 
$(2\ell+1,2)$-problem. Consider a graph $R$ given by 
the adjacency matrix
$
	(M_R)_{i,j} =
	\begin{cases}
		0 & \text{if} \enspace \compatible{\ell}{2}(w^{(i)},w^{(j)}), \\
		1 & \text{otherwise}.
	\end{cases}
$

Then the generating sets are ``maximal'' independent sets (maximal with 
respect to inclusion) in the graph~$R$. We require maximality here since it is already required
by the second property in Definition~\ref{def_gen_set_nemax}.

We can partially simplify the graph~$R$. We say that two vertices $v$ and $w$ in this graph are equivalent if
$
	\forall x\in V \bigl(
		\compatible{\ell}{k}(x,v) \iff
		\compatible{\ell}{k}(x,w)
	\bigr).
$
Then we can put all equivalent vertices into one vertex, i.e., every vertex will contain a set of words instead of only one word. The step with searching ``maximal'' independent sets stays unchanged.

\begin{example}
Let $k=2$. For every $\ell\in\{1,\ldots,4\}$,
all seeds solving the $(\ell+1,2)$-problem are mutually compatible, which means that there exists a unique $(\ell,2)$-generating set. We list them out in the following table.
\[
\begin{array}{l|l}
	\ell & G \\\hline
	1 & \{\J\J\}\\
	2 & \{\F\J\J, \J\F\J, \J\J\F, \J\J\J\}\\
	3 & \{\F\J\J\F,\F\J\J\J, \J\F\J\J, \J\J\F\J, \J\J\J\F, \J\J\J\}\\
	4 & \{
\F\F\J\J\J,
\J\F\F\J\J,
\J\J\F\F\J,
\J\J\J\F\F,
\F\J\F\J\J,
\J\F\J\F\J,
\J\J\F\J\F,
\F\J\J\F\J,
\J\F\J\J\F,
\F\J\J\J\F,\\
&\qquad \F\J\J\J\J,
\J\F\J\J\J,
\J\J\F\J\J,
\J\J\J\F\J,
\J\J\J\J\F,
\J\J\J\J\J
	\}
\end{array}
\]
\end{example}

\begin{example}\label{ex_k2_l5}
Let $k=2$ and $\ell=5$. We find the graph $R$ by the procedure above. After its simplification, we obtain the graph in Figure~\ref{fig_max_independent}, where
\begin{align*}
P_0 &= \{
	\J\J\J\J\J\J\,;\
	\J\J\J\J\J\F\,,\  
	\J\J\J\J\F\J\,,\ 
	\J\J\J\F\J\J\,,\ 
	\J\J\F\J\J\J\,,\ 
	\J\F\J\J\J\J\,,\ 
	\F\J\J\J\J\J\,;\ \\
 &\qquad
	\J\J\J\J\F\F\,,\ 
	\J\J\J\F\F\J\,,\ 
	\J\J\F\F\J\J\,,\ 
	\J\F\F\J\J\J\,,\ 
	\F\F\J\J\J\J\,;\\
 &\qquad
	\J\J\J\F\J\F\,,\ 
	\J\J\F\J\F\J\,,\ 
	\J\F\J\F\J\J\,,\ 
	\F\J\F\J\J\J\,;\\
 &\qquad
	\J\J\F\J\J\F\,,\ 
	\J\F\J\J\F\J\,,\ 
	\F\J\J\F\J\J\,;\ 
	\J\F\J\J\J\F\,,\ 
	\F\J\J\J\F\J\,;\ 
	\\
 &\qquad
	\F\J\J\J\F\F\,;\ 
	\F\F\J\J\J\F\,;\ 
	\F\J\J\J\J\F
 \,\},\\
 P_1 &= \{\F\J\J\F\J\F
 \,\}, \qquad
 P_2 = \{
 	\J\J\F\F\J\F\,,\ 
 	\J\F\F\J\F\J\,,\ 
 	\F\F\J\F\J\J
 \,\},\\
 P_3 &= \{
	\J\F\F\J\J\F\,,\
	\F\F\J\J\F\J
 \,\}, \qquad
 P_4 = \{
 	\J\F\J\J\F\F\,,\ 
 	\F\J\J\F\F\J
 \,\},\\
 P_5 &= \{
 	\J\J\F\J\F\F\,,\
 	\J\F\J\F\F\J\,,\
 	\F\J\F\F\J\J
 \,\}, \qquad P_6 = \{
 	\F\J\F\J\J\F
 \,\}.
\end{align*}

\begin{figure}[t]\centering
\begin{tikzpicture}[
	->,
	>=stealth',
	shorten >=1pt,
	auto,
	node distance=2cm,
	main node/.style={fill=blue!20,draw}
]

  \node[main node] (1)  {$P_1$};
  \node[main node] (2) [right of=1] {$P_2$};
  \node[main node] (3) [below of=2] {$P_3$};
  \node[main node] (4) [right of=2] {$P_5$};
  \node[main node] (5) [below of=4] {$P_4$};
  \node[main node] (6) [right of=4] {$P_6$};
  \node[main node] (0) [left of=3] {$P_0$};

  \path[every node/.style={font=\sffamily\small}]
    (1) edge [-] node [] {} (2)
    (2) edge [-] node [] {} (3)
    (2) edge [-] node [] {} (4)
    (3) edge [-] node [] {} (5)
    (4) edge [-] node [] {} (5)
    (4) edge [-] node [] {} (6);
\end{tikzpicture} 
\caption{The simplified graph of sets of equivalent seeds for $(5,2)$-generating sets search in Example~\ref{ex_k2_l5}.}
\label{fig_max_independent}
\end{figure}

By finding ``maximal'' independent sets in the graph in Figure~\ref{fig_max_independent}, we get all $(5,2)$-generating sets:
\begin{align*}
	G_1 = P_0 \cup P_1 \cup P_3 \cup P_5, 
	&&
	G_2 = P_0 \cup P_1 \cup P_3 \cup P_6, \\
	G_3 = P_0 \cup P_2 \cup P_4 \cup P_6, 
	&&
	G_4 = P_0 \cup P_1 \cup P_4 \cup P_6. 
\end{align*}

\end{example}

To conclude the section, let us remark that a similar
derivation can be done using hypergraphs also for~$k>2$.


\section{Conclusion}

In this paper, we have studied lossless seeds from the perspective
of symbolic dynamics.
We have concentrated on the seed margin~$\ell$ defined as a difference of the length~$m$ of compared strings and the length of a seed.
We have derived asymptotic behavior of~$\ell$ for optimal seeds (Proposition~\ref{prop_asymptotic}), which must satisfy
$\ell \in \Theta(m^{\frac{k}{k+1}}) = \Theta(m-w(m))$.
We have shown another criterion for errors detection by seeds (Theorem~\ref{thm_detection}).
From this criterion we have proved that lossless seeds coincide with languages of certain sofic subshifts, therefore, they are recognized by finite automata (Theorem~\ref{thm_S_k(l)_regular}).
We have presented that these subshifts are fully
given by the number of allowed errors~$k$ and
the seed margin~$\ell$ and that they 
can be further decomposed into subshifts of finite type.

These facts explain why periodically repeated patterns often appear in lossless seeds. 
This is caused by the fact that
these patterns correspond to cycles in recognizing automata (which 
correspond to seeds for cyclic $(m,k)$-problems
in~\cite{multiseed}).
Nevertheless, it remains unclear what is the upper bound on the length of cycles to obtain at least some optimal seeds.
In the case case $k=2$, it was conjectured
in~\cite[Conjecture~1]{diplomka}
that
it is sufficient to consider patterns having length at most $\ell+1$ to obtain some of optimal seeds.

\medskip

\textbf{Acknowledgements.}
The author is supported by the ABS4NGS grant of the French government (program \emph{Investissement d'Avenir}). He is grateful to Gregory Kucherov and Karel Klouda for helpful ideas.
He also thanks three anonymous referees for valuable comments.


\nocite{*}
\bibliographystyle{eptcs}
\bibliography{languages_of_lossless_seeds}
\end{document}